\documentclass[11pt]{article}

 \usepackage{mst-stylefile}
 \usepackage{Bbb11a}
 \usepackage{harvard}
 \usepackage{amsfonts}
 \usepackage{amssymb}
 \usepackage{mathrsfs}

 \usepackage{color}

\include{epsf}

\renewcommand\cal{\mathcal }

\renewcommand{\Bbb}{\mathbb}

\newcommand{\bbR}{{\Bbb R}}

\newcommand{\bbN}{{\Bbb N}}

\newcommand\E{\Bbb E}
\renewcommand\P{\Bbb P}
\newcommand\noi{\noindent}

\newcommand{\what}{\widehat}

\def\Sum{\mathop{\sum}}

\def\wtilde{\widetilde }

\renewcommand{\cite}{\citeyear}

 \title{Max--stable sketches: estimation of $\ell_\alpha-$norms, dominance norms and point 
 queries for non--negative signals}

 \date{February 6, 2006}
 \author{Stilian A.\ Stoev \\ Department of Statistics\\
 {University of Michigan, Ann Arbor} \\ {{\tt sstoev@umich.edu} }
 \\ {\ }\\ Murad S.\ Taqqu \\ {Department of Mathematics and Statistics} \\ 
 {Boston University} \\ {\tt murad@bu.edu}}

\begin{document}

 \maketitle

 \begin{abstract}
 Let $f:\{1,2,\ldots,N\}\to [0,\infty)$ be a non--negative signal, defined over a 
 very large domain and suppose that we want to be able to address
 approximate aggregate queries or point queries about $f$. To answer
 queries about $f$,  we introduce a new type of random sketches
 called {\it max--stable sketches}.  The (ideal precision) max--stable sketch
 of $f$, $E_j(f),\ 1\le j\le K$, is defined as:
 $$
  E_j(f) := \max_{1\le i \le N} f(i) Z_j(i),\ \ 1\le j \le K,
 $$
 where the $KN$ random variables $Z_j(i)$'s are independent with 
 standard $\alpha-$Fr\'echet distribution, that is, $\P\{ Z_j(i) \le x\}
 = \exp\{-x^{-\alpha}\},\ x>0$, where $\alpha$ is an arbitrary positive
 parameter.  Max--stable sketches are particularly natural when dealing with 
 maximally updated data streams, logs of record events and dominance norms or
 relations between large signals. 
 By using only max--stable sketches of relatively small size $K<<N$, we can compute
 in small space and time: {\it (i)} the $\ell_\alpha-$norm, $\alpha>0$, of the signal
 {\it (ii)} the distance between two signals in a metric, related to the
 $\ell_\alpha-$norm, and {\it (iii)} dominance norms, that is, the norm of the maxima
 of several signals.  In addition, we can also derive point queries about the signal.

 As is the case of $p-$stable, $0<p\le 2$, (sum--stable) additive
 sketches, see Indyk \cite{indyk:2000}, max--stability ensures that
 $E_j(f) \stackrel{d}{=} \|f\|_{\ell_\alpha} Z_1(1),$ with
 $\|f\|_{\ell_\alpha} = (\sum_{1\le i\le N} f(i)^\alpha)^{1/\alpha},$
 where $\stackrel{d}{=}$ means equal in distribution.  
 We derive $\epsilon-\delta$ probability bounds on the relative
 error for distances and dominance norms.  This can be implemented by
 efficient algorithms requiring small space even when the computational
 precision is finite and a limited amount of random bits are used.
 Our approach in approximating dominance norms improves considerably on
 existing techniques in the literature.
 \end{abstract}

 \section{Introduction and motivation}

 Random sketches have become an important tool in building unusually
efficient algorithms for approximate representation of large data
sets. One of their major applications is to data streams.  To put our
work into perspective, we start by describing briefly the data
streaming context.  We then list some of the major contributions and
discuss our results.

Consider an integer--valued signal $f:\{1,\ldots,N\}\to
\{-M,\ldots,M\}$, defined over a ``very large'' domain, so that it is
not feasible to store and/or process it in real time.  The signal is
updated or acquired sequentially in time, starting with the zero
signal at time zero.  Following, for example, Gilbert {\it et al.\ }
\cite{gilbert:kotidis:muthukrishnan:strauss:2001} (see also
Muthukrishnan \cite{muthukrishnan:2003} and the seminal work of Henzinger, Raghavan
and Rajagopalan \cite{henzinger:raghavan:rajagopalan:1998}), we focus on two streaming
models: {\it (i)} {\it cash register}, and {\it (ii)} {\it aggregate}.
In case {\it (i)}, data pairs $(i,a(i))$ are observed successively (in
arbitrary order in $i$) and on each data arrival, the $i-$th component
of the signal is updated incrementally (like a cash register): $f(i):=
f(i) + a(i)$.  In case {\it (ii)}, the data pairs
$(i, f(i))$ are observed directly (again in arbitrary order in $i$).  In this case,
a given index $i$ appears at most once and the corresponding $f(i)$'s are
not updated incrementally multiple times.  Model {\it (i)} is more general and more widely
used. Both models, however, have found important applications in many
areas such as on-line processing of large data bases, network traffic
monitoring, computational geometry, etc.\ (see, e.g.\ Gilbert,
Kotidis, Muthukrishnan and Strauss
\cite{gilbert:kotidis:muthukrishnan:strauss:2002,gilbert:kotidis:muthukrishnan:strauss:2001Q},
Cormode and Muthukrishnan \cite{cormode:muthukrishnan:2003W}, Indyk
\cite{indyk:2003}).  Much of the work on sketches was motivated by the
seminal paper of Alon, Matias and Szegedy
\cite{alon:matias:szegedy:1996}.  For a detailed review of methodologies and
applications in this emerging area in theoretical computer science
see  Muthukrishnan
\cite{muthukrishnan:2003}.

 {\it Random sketches} are statistical summaries of the signal $f$,
which can be updated sequentially (as the stream is observed) using
little 
processing time, processing space and computations.  Many algorithms involving
random sketches have been proposed, see e.g.\ Muthukrishnan
\cite{muthukrishnan:2003} and the references therein.  They provide as
a common
feature,  approximations to various queries
(functionals) on the signal $f$, within a factor of $(1\pm \epsilon)$,
with probability at least $(1-\delta)$, where $\epsilon >0$ and
$\delta>0$ are ``small'' error and probability parameters chosen in
advance.  Typically, this is realized by algorithms consuming an amount
$$
{\cal O}{\Big(}\log_2 M (\log_2 N)^{{\cal O}(1)} \ln(1/\delta) /
\epsilon^{{\cal O}(1)}{\Big)}
$$
 of storage, and even smaller order of processing time per stream
item $f$. 
 Here $M$ denotes the size of the range of the values $a(i)$
 in the cash register model  or
$f(i)$ in the aggregate model.  In
many applications, one may be willing to  sacrifice deterministic approximations at
the expense of stochastic approximations, which are valid with high
probability and are easy to compute.

Indyk \cite{indyk:2000} has pioneered the use of $p-$stable
distributions in random sketches (see also, Feingenbaum, Kannan, Strauss and Viswanathan \cite{feingenbaum:kannan:strauss:viswanathan:1999},
Fong and Strauss \cite{fong:strauss:2000}, Cormode
\cite{cormode:2003} and Cormode and Muthukrishnan
\cite{cormode:muthukrishnan:2003}).  The $p-$stable ($0<p\le 2$)
sketch of the signal $f$ is defined as:
 \beq
\label{e:Sj}
 S_j(f) :=
\sum_{i=1}^N f(i) X_j(i),\ \ j=1,\ldots,K, 
\eeq
 where the $KN$ random variables $X_j(i)$
are independent with a  $p-$stable distribution.
  The stability (sum--stability, see the Appendix) property of the
  $p-$stable distribution
implies that
 \beq
\label{e:Sj-stable}
 S_j(f) \stackrel{d}{=}
{\Big(}\sum_{i=1}^N |f(i)|^p {\Big)}^{1/p} X_1(1) = \|f\|_{\ell_p}
X_1(1),\ \ j=1,\ldots, K,
 \eeq
 where $\stackrel{d}{=}$ means equal in
distribution.  Also, by the linearity of the inner product
\refeq{Sj}, the sketch $S_j(f),\ j=1,\ldots,K$ of $f$ can be updated
sequentially, in both, {\it cash register} and {\it aggregate}
streaming models with ${\cal O}(K)$ operations per pair $(i,a(i))$ or
$(i,f(i))$, where the $X_j(i)$'s are generated efficiently, on demand
(see below).  In his seminal work, Indyk \cite{indyk:2000}, used $p=1$
(Cauchy distributions for the $X_j(i)$'s) and the median statistic
$$
  {\rm median}\{|S_j(f)|,\ {1\le j \le K}\}
$$ 
to estimate the norm $\|f\|_{\ell_1} := \sum_{i=1}^N |f(i)|$ of the
signal.  It was shown that for any $\epsilon>0$ and $\delta>0$, the
norm $\|f\|_{\ell_1}$ is estimated within a factor of $(1\pm
\epsilon)$ with probability at least $(1-\delta)$, provided that
\beq\label{e:Indyk} K \ge {\cal O} {\Big(} \frac{1}{\epsilon^2}
\log(1/\delta) {\Big)}.  \eeq Moreover, by using the results of Nisan
\cite{nisan:1990}, these estimates were shown to be realized with
${\cal O}(\log_2 M \log_2 (N/\delta) \log(1/\delta)/\epsilon^2)$ bits
of storage, needed primarily to store {\it truly random bits} or {\it
seeds} for the pseudorandom generator.  Roughly speaking, these seed
bits are used to efficiently generate any one of the $K N$ variables
$Z_j(i)$'s on demand, when a data pair $(i,a(i))$ or $(i,f(i))$ is
observed.

Exploiting the linearity of sketches and the properties of the stable
distributions, Indyk \cite{indyk:2000} also developed approximate
embeddings of high--dimensional signals $f \in \ell_1^N$ in
$\ell_1^m$, where $m<< N$.  This allowed efficient approximate
solutions to difficult nearest neighbor search algorithms in high
dimensions, see e.g.\ Datar {\it et al.}
\cite{datar:immorlica:indyk:mirrokni:2004}.  Other authors also used
stable distributions to construct efficient stochastic approximation
algorithms, see e.g.\ Cormode and Muthukrishnan
\cite{cormode:muthukrishnan:2003}.

Here, we propose a novel type of random sketches, called {\it max--stable sketches}.  Namely, consider 
a {\it non--negative} signal $f:\{1,\ldots,N\} \to [0,\infty)$.  The $\alpha-$max--stable sketch of
$f$ is defined as:
\beq\label{e:Ej}
 E_j(f) := \max_{1\le i \le N} f(i) Z_j(i),\ \ j=1,\ldots,K,
\eeq
where the $KN$ random variables $Z_j(i)$
 are independent standard $\alpha-$Fr\'echet , that is, 
\beq\label{e:Frechet}
 \P\{ Z_j(i) \le x\} = \Phi_\alpha(x) := \left\{\begin{array}{ll}
                               \exp\{ -x^{-\alpha}\} &,\ x>0\\
                                0 &,\ x\le 0,
                              \end{array} \right.
\eeq
and where $\alpha>0$ is an arbitrary positive parameter.  

The max--stable sketches can only be maintained in the {\it aggregate}
streaming model, that is, when any given index value $i$ is observed at
most once in the stream. This is so because the operation ``max'' is not linear. 
Indeed, if the signal values $f(i)$ were incremented sequentially (as in the cash register
model), then to be able to update the max--stable sketch of $f$, one would have to know
the whole signal thus far, which is not feasible.  Other than the {\it aggregate} model, a
natural streaming context for max--stable sketches is when the {\it cash register} is updated in a
max--incremental fashion:
$$
 f(i) := \max\{f(i),a(i)\}.
$$
In this setting, max--stable sketches can be maintained sequentially.  

\medskip
In the spirit of $p-$stable sketches, the max--stability of the $Z_j(i)$'s implies (see the Appendix):
$$
  E_j (f) \stackrel{d}{=} {\Big(}\sum_{i=1}^N f(i)^\alpha {\Big)}^{1/\alpha} Z_1(1)
                       = \|f\|_{\ell_\alpha} Z_1(1),\ \ j=1,\ldots,K.
$$ 
Therefore, as in Indyk \cite{indyk:2000},
 for any $\epsilon>0$ and $\delta>0$, we can estimate the norm
$\|f\|_{\ell_\alpha}$ within a factor of $(1\pm \epsilon)$,
 with probability at least $(1-\delta)$, if
$$
 K\ge {\cal O}{\Big(}\frac{1}{\epsilon^2}\ln(1/\delta){\Big)}.
$$
Following the ideas in Indyk \cite{indyk:2000}, we show by using results of Nisan \cite{nisan:1990}, 
that this can be realized with real algorithms of space 
$$
 {\cal O} {\Big(}\log_2(M)\log_2(N/\delta)\ln(1/\delta)/\epsilon^2{\Big)},\  \mbox{ and }\ 
 {\cal O} {\Big(}\log_2(M)\ln(1/\delta)/\epsilon^2{\Big)}
$$
processing time per stream item $(i,f(i))$.  Note that $\alpha>0$ can be chosen arbitrarily large, 
whereas one is limited to $0<p\le 2$ in the $p-$stable case.  Since the max--stable
sketches are non--linear, being able to approximate $\|f\|_{\ell_\alpha},$ for any
$\alpha>0$, does not imply approximation of the distance
$\|f-g\|_{\ell_\alpha}$
 in the $\ell_\alpha-$norm of
two signals based on their sketches.  Therefore, our results do not contradict the findings of
Saks and Sun \cite{saks:sun:2002}.  The recent paper of Indyk and Woodruff \cite{indyk:woodruff:2005} provides
algorithms for approximating $\ell_\alpha-$norms for $\alpha>2$ which essentially match the theoretical lower
bounds on the complexity in Saks and Sun \cite{saks:sun:2002}.  The strengths of max--stable sketches lie in 
approximating max--linear functionals.

\medskip
{\it  One of the key advantages of max--stable sketches is in handling dominance norms}.    
Cormode and Muthukrishnan \cite{cormode:muthukrishnan:2003}, consider the problem of estimating the norm
of the dominant of several signals, that is, {\it dominance norms}.  Given non--negative signals 
$f_r(i),\ 1\le i \le N,$ $r=1,\ldots,R$, the goal is to estimate the norm $\|f^*\|_{\ell_\alpha}$, where
\beq\label{e:f-star}
 f^*(i) := \max_{1\le r\le R} f_r(i),\ \ 1\le i \le N.
\eeq
Such type of problems are of interest when monitoring Internet traffic, for example, where $f_r(i)$ stands for the number of
packets transmitted by IP address $i$ in its $r-$th transmission.  The signal $f^*$ then represents {\it worst case scenario}
in terms of traffic load on the network and its norm or various other characteristics are of interest to network administrators.
Other applications of dominance norms arise when studying electric grid loads and in finance (for more details, 
see Cormode and Muthukrishnan \cite{cormode:muthukrishnan:2003} and the references therein).
A novel area of applications of max--stable sketches arises in privacy, see  
Ishai, Malkin, Strauss and Wright \cite{ishai:malkin:strauss:wright:2005}.

Suppose that we only have access to the max--stable sketches $E_j(f_r),\ 1\le j\le K$ of the signals $f_r$ in \refeq{f-star}.
In view of {\it max--linearity}, we then have
$$
 E_j(f^*) = \max_{1\le r \le R} E_j(f_r),\ \ 1\le j \le K.
$$
That is, one can recover the max--stable sketch of the dominant signal $f^*$ by taking component--wise maxima of the
sketches of the signals $f_r,\ 1\le r\le R$.  Therefore, the quantity $\|f^*\|_\alpha$, which is the dominance norm of 
the signals $f_r,\ 1\le r\le R$, can be readily estimated from the sketch $E_j(f^*)$ by using medians or sample moments.
Moreover, this can be done with precision within a factor of $(1\pm\epsilon)$ and with probability at least $(1-\delta)$,
provided that $K \ge {\cal O} (\ln(1/\delta)/\epsilon^2)$.  In practice, one deals with finite precision calculations and
pseudo--random number generators of bounded space.  In this setting, as in the case of approximating plain norms, one
can compute dominance norms by using an algorithm with processing space 
${\cal O} (\log_2(M)\log_2(N/\delta)\ln(1/\delta)/\epsilon^2)$, and 
${\cal O} (\log_2(M)\ln(1/\delta)/\epsilon^2)$ per item processing time. 
Our approach consumes less randomness, space and improves significantly on the processing time in the method of
Cormode and Muthukrishnan \cite{cormode:muthukrishnan:2003} (see Theorem 2 therein).  

In addition to norms and dominance norms, one can use max--stable sketches to recover large values of the
signal exactly.  We construct a point estimate $\what f(i_0)$ of $f(i_0)$, $i_0\in\{1,\ldots,N\}$, based 
on an {\it ideal precision} $\alpha-$max--stable sketch of size $K$, such that for any $\epsilon>0$ and $\delta>0$, 
\beq\label{e:fhat}
 \P \{ \what f(i_0) = f(i_0) \} \ge 1-\delta,\ \mbox{ if }\ f(i_0) > \epsilon \| f\|_{\ell_\alpha}
 \ \ \mbox{ and }\ \ K \ge { \ln(1/\delta) \over \epsilon ^\alpha}.
\eeq
Real algorithms of space
$$
 {\cal O}(\log_2(DN/\delta\epsilon)\log_2(N/\delta)\log_2(1/\delta)/\epsilon^\alpha) = {\cal O}((\log_2(DN/\delta\epsilon))^{{\cal O}(1)}
\log_2(1/\delta)/\epsilon^\alpha)$$
and smaller order of per stream item processing time exist.  
Here $D= 1+ f(i_0)/(\sum_{i\not=i_0}f(i)^\alpha)^{1/\alpha}$ and the signal $f$ is
supposed to take integer values.  Observe that one can easily maintain the largest $1/\epsilon^\alpha$ values of the signal
exactly.  Although, our method does not improve on the naive approach, the proposed estimator may be useful when the signal
is not directly observable but its max--stable sketch is available.  This is particularly useful in 
applications related to privacy, see the forthcoming paper of 
Ishai, Malkin, Strauss and Wright \cite{ishai:malkin:strauss:wright:2005}.

\medskip
Important ideas exploiting min--stability have been successfully used in the literature.
Cohen \cite{cohen.edith:1997} assigns to the items of a positive signal independent Exponential variables with parameters equal
to the signal values.  The minima of independent exponentials is an exponential variable with parameter
equal to the sum of the parameters of the individual components.  Therefore, by keeping only the minima of such exponential
variables corresponding to certain ranges of the signal values, one can estimate the sum of the signal values in these ranges.
This can be done efficiently, in small space and time, by taking independent copies of  such minima, see Theorem 2.3 in 
Cohen \cite{cohen.edith:1997}.  This approach can be viewed as a dual approach to that of the max--stable sketches.  Indeed, the
reciprocal of an Exponential variable is $\alpha-$Fr\'echet with $\alpha=1$.  We provide here a more general framework where
$\alpha$ can be arbitrary positive parameter and therefore we can estimate not only sums but $\ell_\alpha-$norms.  In fact, 
going a step further, we estimate efficiently {\it dominance norms} of several signals and show that relatively 
large values can be recovered exactly with high probability.

Alon, Duffield, Lund and Thorup \cite{alon:duffield:lund:thorup:2005} suggest an interesting random priority sampling scheme.
It assigns random priority $q_i = w_i/U_i,$ to an item $i$ which has weight $w_i>0$, where the $U_i$'s are independent uniformly 
distributed random numbers in $(0,1)$.  In our scenario $w_i$ corresponds to the signal value $f(i)$.  However, instead of
taking maxima of the priorities $q_i$ over $i$, these authors consider the
top$-k$ largest priority items.  By using a statistic, relative to the $(k+1)-$st largest priority, they can estimate
efficiently and with high probability the sum of the weights $w_i$ for relatively small $k$'s.  This is an interesting
approach and it differs from that of the max--stable sketches in two major aspects: {\it (i)} the random priorities $q_i$ have
Pareto distribution with heavy--tail exponent $1$, whereas we employ Fr\'echet distributions to be able to use their
max--stability property; {\it (ii)} in priority sampling, one keeps the top$-k$ values, whereas max--stable sketches keep 
different realizations of the maximal ``priority''.  The second difference between max--stable sketches and the priority 
sampling scheme of Alon, Duffield, Lund and Thorup 
\cite{alon:duffield:lund:thorup:2005} is crucial since the top$-k$ priorities are {\it dependent} random variables.  This fact,
we believe, makes the rigorous analysis of the variance in the priority sampling scheme rather difficult
(see Conjecture 1 in the last reference).  Nevertheless, priority sampling involves generating only $N$ independent
realizations, where $N$ is the size of the signal. It is thus computationally more efficient than our method and the method of Cohen
 \cite{cohen.edith:1997} (see also Section 1.5.6 in Alon, Duffield, Lund and Thorup 
\cite{alon:duffield:lund:thorup:2005} for a discussion).  

In summary, the max--stable sketches proposed here are natural when dealing with {\it dominance norms} and
$\ell_\alpha-$norms for arbitrarily large $\alpha>0$.  Their properties, moreover, can be established precisely and
rigorously related to the nature of the signal $f$.  Max--stable sketches complement and improve on existing techniques, and
can offer a new ``non--linear'' dimension in stochastic approximation algorithms.

\section{Approximating  $\ell_\alpha-$norms and distances}
 \label{s:dist}

We  show here that max--stable sketches can be used to estimate norms and certain distances between  
two signals.  For simplicity, we deal here with ideal precision sketches.  In an extended version of the
paper we show that efficient real algorithms exist by using the results of Nisan \cite{nisan:1990}.

 We first focus on estimating the norm $\|f\|_{\ell_\alpha} = (\sum_i f(i)^\alpha)^{1/\alpha}$ from
 the $\alpha-$max--stable sketch $E_j(f),\ 1\le j\le K$ of $f$  (see \refeq{Ej}).  Introduce the quantities
 $$
  \ell_{\alpha,r}(f) := {\Big(} \frac{1}{\Gamma(1-r/\alpha) K} \Sum_{j=1}^K E_j(f)^r {\Big)}^{1/r}
 $$
 for some $0<r<\alpha$, and
 $$
  \ell_{\alpha,{\rm med}}(f) := (\ln 2)^{1/\alpha} {\rm median}\{E_j(f),\ 1\le j \le K \},
 $$
 see \refeq{Zp} and \refeq{medZ} below for motivation.   By using the max--stability property of 
 $\alpha-$Fr\'echet distributions, we obtain the following result.

\begin{theorem}\label{t:norm}
 Let $\epsilon>0$ and $\delta>0$.  Then, for all $0<r<\alpha/2$, we have
 $$
 \P\{ |\ell_{\alpha,r}(f)/\|f\|_{\ell_\alpha} - 1| \le \epsilon \} \ge 1-\delta,
\ \ \mbox{ and }\ \ \P\{ |\ell_{\alpha,{\rm med}}(f)/\|f\|_{\ell_\alpha} - 1| \le \epsilon \}\ge 1-\delta,
$$
provided that $K\ge C \log(1/\delta)/\epsilon^2$. Here the constant $C>0$ depends only on 
$\alpha$ and $r$, in the case of the $\ell_{\alpha,r}(f)$ estimator.
\end{theorem}

\noi The idea of the proof is given in the Appendix.  It relies on the fact that $E_j(f)^r$ have
finite variances for all $0<r<\alpha/2$ and uses the Central Limit Theorem. More general results
where $\alpha/2 \le r < \alpha$ will be given in an extended version of the paper.

\medskip
The last result indicates that the quantities $\ell_{\alpha,r}(f),\ 0<r<\alpha/2$ and $\ell_{\alpha,{\rm med}}$
approximate $\|f\|_{\ell_\alpha}$ up to a factor of $(1\pm\epsilon)$ with probability at least $(1-\delta)$.
This can be achieved with an ideal--precision sketch of size $K = {\cal O}(\log(1/\delta)/\epsilon^2)$.

\bigskip
{\it We now focus on approximating distances.}  Consider two signals $f,\ g:\{1,\ldots,N\} \to [0,\infty)$ and
let $E_j(f),\ E_j(g),\ j=1,\ldots,K$ be their {\it ideal precision} max--stable 
sketches.  Observe that the max--stable sketches are non--linear and therefore
even if $f(i)\le g(i),\ 1\le i \le N$, the sketch $E_j(g-f)$ does not equal $E_j(g)-E_j(f)$.  Nevertheless, one
can introduce a distance between the signals $f$ and $g$, other than the norm $\|f-g\|_{\ell_\alpha}$ 
which can be computed by using the sketches $E_j(f)$ and $E_j(g)$.
 
Consider the functional 
$$
 \rho_\alpha(f,g):=
  \|f^\alpha - g^\alpha\|_{\ell_1} = \Sum_{i} |f(i)^\alpha - g(i)^\alpha|.
$$
One can verify that $\rho_\alpha(f,g)$ is a metric on $\bbR_+^N$.  This metric, rather than the
norm $\|f-g\|_{\ell_\alpha}$, is more natural when dealing with max--stable sketches 
(see, Stoev and Taqqu \cite{stoev:taqqu:2005}).  

Observe that 
 \begin{eqnarray*}
  \|f^\alpha - g^\alpha\|_{\ell_1}& =& \Sum_{i} (f(i)^\alpha\vee g(i)^\alpha - f(i)^\alpha) 
  + \Sum_{i}(f(i)^\alpha\vee g(i)^\alpha - g(i)^\alpha)\\
 & =& 2 \| f\vee g\|_{\ell_\alpha}^\alpha - \|f\|_{\ell_\alpha}^\alpha
 -\|g\|_{\ell_\alpha}^\alpha.
\end{eqnarray*}
By the max--linearity of max--stable sketches, we get $\E_j(f\vee g) = E_j(f)\vee E_j(g)$ (see \refeq{max-linear}, below).
Therefore, the terms in the last expression can be estimated in terms the estimators $\ell_{\alpha,r}(f)$ 
and $\ell_{\alpha,{\rm med}}(f)$ above.  Namely, we define
$$
\what\rho_{\alpha,r}(f,g) := 2 \ell_{\alpha,r}(f\vee g)^\alpha - \ell_{\alpha,r}(f)^\alpha - \ell_{\alpha,r}(g)^\alpha,
$$
for some $0<r<\alpha$, and
$$
\what\rho_{\alpha,{\rm med}}(f,g) := 2 \ell_{\alpha,{\rm med}}(f\vee g)^\alpha - \ell_{\alpha,{\rm med}}(f)^\alpha
  - \ell_{\alpha,{\rm med}}(g)^\alpha.
$$

\begin{theorem}\label{t:dist} Let $\epsilon>0$, $\delta>0$ and $\eta>0$. If 
\beq\label{e:separation}
 \rho_\alpha(f,g) \ge \eta \|f \vee g \|_{\ell_\alpha}^\alpha,
\eeq
then, for all $0<\rho <\alpha/2$, we have
\beq\label{e:dist}
 \P{\Big\{}|\what \rho (f,g) /\rho(f,g) - 1 | \le {\cal O}(\epsilon/\eta) {\Big \}} \ge 1-3\delta,
\eeq
provided that $K \ge C \ln(1/\delta)/\epsilon^2$. Here the constant $C>0$ depends only on 
$\alpha$ and $r$, in the case of the $\what\rho_{\alpha,r}(f,g)$ estimator.
Here $\what \rho(f,g)$ stands for either $\what\rho_{\alpha,{\rm med}}(f,g)$ or $\what\rho_{\alpha,r}(f,g)$.
\end{theorem}

\noi The idea of the proof is given in the Appendix.  The condition \refeq{separation} is essential.  Indeed, 
by taking indicator signals $f(i) = 1_{A}(i)$ and $g(i) = 1_B(i)$, we get
that 
$$
 \rho_\alpha(f,g) = \Sum_{i} 1_A(i) 1_B(i) = |A\cap B|.
$$
Therefore, if condition of type  \refeq{separation} was not present, one would be able to efficiently
estimate the intersection of the two sets $A$ and $B$ with small relative error, which is proved to be 
a hard problem (see, Razborov \cite{razborov:1992} and also Section 4 in
Cormode and Muthukrishnan \cite{cormode:muthukrishnan:2003}).

\section{Approximating dominance norms}

\label{s:dom}

 Let now $f_r(i),\ 1\le i \le N,$ $r=1,\ldots,R$ be $R$ non--negative signals defined over large domain.
Our goal is to approximate their dominance $\ell_\alpha-$norm, that is, the norm $\|f^*\|_{\ell_\alpha},\ \alpha>0,$ 
of the signal
$$
 f^*(i) := \max_{1\le r \le R} f_r(i), \ \ 1\le i \le N.
$$

As argued in the introduction, such problems arise in Internet traffic monitoring, electric grid management and 
also in finance.  The seminal paper of Cormode and Muthukrishnan \cite{cormode:muthukrishnan:2003} addresses 
the problem of dominance norm estimation in the special case $\alpha=1$.  It was shown therein that the problem
has a small space and time approximate solution, valid with high probability.  The main tool used used in the last work
are $p-$(sum)stable sketches of the data where the stability index $p>0$ is taken to have very small values.
Here, we propose an alternative solution to the dominance norm problem, which is superior in terms of time and space
consumption and also works when dealing with $\|\cdot\|_{\ell_\alpha}$ for an arbitrary $\alpha>0$.  In the end of this 
section, we also show the connection between our approach and that of Cormode and Muthukrishnan \cite{cormode:muthukrishnan:2003}.

Let  $E_j(f_r),\ 1\le j \le K$ be the $\alpha-$max--stable sketches of the signals $f_r,\ r=1,\ldots,R$.
By  max--linearity of the max--stable sketch:
\beq\label{e:f*}
 E_j( f^*) = \max_{1\le r\le R} E_j(f_r),\ \ \forall j,
\eeq
where $f^* (i) = \max_{1\le r \le R} f_r(i)$, $1\le i\le N$.

\noi Hence, by using sample medians for example, we get
\begin{eqnarray*}
& & {\rm dom}_{\rm max,\alpha}(f_1,\ldots,f_R) := \|f^*\|_{\ell_\alpha} \approx \ell_{\alpha,{\rm med}}(f^*)\\
& & \ \ \ \ \ \    
 = (\ln 2)^{1/\alpha} {\rm median} \{ \vee_r E_j(f_r),\ 1\le j \le K\}.
\end{eqnarray*}

\noi Our first results on $\ell_\alpha-$norm approximation imply:

\begin{theorem}\label{t:dom} Let $\epsilon>0$ and $\delta>0$.  For all $0<r<\alpha/2$:
$$
 \P\{ |\ell_{\alpha,r}(f^*)/\|f^*\|_{\ell_\alpha} - 1| \le \epsilon \}\ge 1-\delta,\ \ \mbox{ and }\ \ 
  \P\{ |\ell_{\alpha,{\rm med}}(f^*)/\|f\|_{\ell_\alpha} - 1| \le \epsilon \}\ge 1-\delta,
$$
provided  
$$
  K\ge C \log(1/\delta)/\epsilon^2.
$$
Here the constant $C>0$ depends only on $\alpha$ and $r$, in the case of the $\ell_{\alpha,r}(f^*)$ estimator.
\end{theorem}

\noi The proof of this result follows from the max--linearity of the max--stable sketches (see \refeq{f*}) and
Theorem \ref{t:norm} above.  We now make some remarks on the differences between our approach and that of Cormode and 
Muthukrishnan \cite{cormode:muthukrishnan:2003}.

\medskip
\noi{\bf Remarks:}
 \begin{itemize}
 \item  Note that the statement of Lemma 1 in the last reference is mathematically incorrect.  One cannot have $\alpha$ 
 in the right--hand side since limit has been taken as $\alpha \to 0+$ in the left--hand side therein.  The correct 
 statement is as follows:

{\it  Let $\xi_\alpha$ be symmetric $\alpha-$stable random variables with constant scale coefficients
 $\sigma>0$.  Then, as $\alpha \to 0+$, we have
 $$
 |\xi_\alpha|^\alpha \stackrel{d}{\longrightarrow} Z,
 $$
 where $Z$ is a standard $1-$Fr\'echet random variable, that is, $\P\{Z\le x\} = \exp\{-1/x\},\ x>0$.
 Observe that $Z = 1/X$, where $X$ is an Exponential random variable with mean $1$.
 See  Exercise 1.29, p.\ 54 in Samorodnitsky and Taqqu \cite{samorodnitsky:taqqu:1994book}.}

 Therefore, the continuity of the cumulative distribution function of the limit $Z$ implies that the medians of
the distributions of $|\xi_\alpha|^\alpha$ converge, as $\alpha\to 0 +$, to the median of $Z$ which is $1/\ln(2)$
(see \refeq{medZ} below).  (Note that here we use 
$\alpha$ as in Cormode and Muthukrishnan \cite{cormode:muthukrishnan:2003} whereas the parameter
$\alpha$ plays a different role in Relation \refeq{medZ} below.)

 \item The method of Cormode and Muthukrishnan \cite{cormode:muthukrishnan:2003} uses $p-$(sum)stable sketches with
 very small $p>0$.  The $p-$stable distributions involved in these sketches have infinite moments of all orders 
 greater than $p$ and in practice take extremely large values.  This poses a number of practical challenges in
 storing and in fact precisely generating these random sketches.  Our method does involve heavy--tailed random 
 variables but they are not as {\it extremely heavy--tailed} and have good computational properties. 
 Furthermore, Fr\'echet distributions can be simulated more efficiently than sum--stable distributions
(see the Appendix).  Therefore, in practice we expect our method to be more robust than the one of Cormode
 and Muthukrishnan \cite{cormode:muthukrishnan:2003}.

 \item The storage and per item processing times of our method are significantly less than those of Cormode
 and Muthukrishnan \cite{cormode:muthukrishnan:2003}.
 \end{itemize}

\section{Answering point queries with max--stable sketches}
\label{s:point}

 Max--stable sketches can  be also used to recover relatively large values of the
 signal exactly with high probability.  This problem has in fact a deterministic solution by using
 a naive algorithm in small space and time.  Namely, as the signal is being observed (in the aggregate or 
 time series model) we simply maintain a vector of the top$-K$ largest values. Max--stable sketches however, can
 be very helpful if no direct access to the signal is available either due to security, computational, power or 
 privacy restrictions (see Ishai, Malkin, Strauss and Wright \cite{ishai:malkin:strauss:wright:2005}).  
 
 We first present the main ideas using  ideal algorithms which assume infinite precision and random variables which can
 be perfectly generated.  We then remove these idealizations.   Consider a non--negative signal
 $f:\{1,\ldots, N\} \to [0,\infty)$. Let $\alpha>0$ and let 
 $E_j(f),\ j=1,\ldots, K$ be an {\it ideal} $\alpha-$max--stable
 sketch of $f$ defined in \refeq{Ej}. 
 Given an $i_0\in \{1,\ldots,N\}$, set
 \beq\label{e:gk}
  g_j(i_0) := { E_j(f) \over Z_j(i_0)}
            ={ \max_{1\le i \le N} f(i) Z_j(i) \over Z_j(i_0)},\ \ j=1,\ldots, K,
 \eeq
 and define the {\it point query} estimate $\what f(i_0)$ as the smallest
 of the $g_{j}(i_0),\ j=1,\ldots,K$. If
 $g_{(j)}(i_0),\ j=1,\ldots,K$ denote the sorted $g_j(i_0)$'s:
$$
  g_{(1)}(i_0) \le g_{(2)}(i_0) \le \cdots \le g_{(K)}(i_0),
$$
then  the {\it point query} estimate $\what f(i_0)$ is
\beq\label{e:fhat-def}
 \what f(i_0) := g_{(1)}(i_0) = \min_{1\le j \le K} g_j (i_0).
\eeq

We  also introduce the following criterion:
$$
 {\rm criterion}(i_0) := \left\{\begin{array}{ll}
                                   1 &,\ \mbox{ if }g_{(1)}(i_0) = g_{(2)}(i_0) \\
                                   0 &,\ \mbox{ if }g_{(1)}(i_0) < g_{(2)}(i_0).
                                 \end{array}
 \right.
$$
which  serves as a proxy for $\what f(i_0) = f(i_0)$, as indicated in
the following theorem.

\begin{theorem} \label{t:pointwise}
 Let $\epsilon\in (0,1)$, $\delta>0$ and $i_0\in\{1,\ldots,N\}$.

  {\bf (i)} If $f(i_0)> \epsilon\|f\|_{\ell_\alpha}$ and $K\ge \ln(1/\delta)/\epsilon^{\alpha}$ 
   (see \refeq{fhat}), then
   \beq\label{e:t:pointwise-i}
    \P\{ \what f(i_0) = f(i_0) \} \ge 1-\delta.
   \eeq

  {\bf (ii)} {\bf (a)}
 ${\rm criterion}(i_0) = 1$ implies $\what f(i_0)  =
 f(i_0)$.

{\bf (ii)} {\bf (b)}
 If for some $\theta > 0$, 
 $$ f(i_0)> \epsilon\|f\|_{\ell_\alpha} \ \mbox{ and }\ 
 K\ge \max\{3, 2C_\theta \ln(2/\delta)/\epsilon^{\alpha + \theta} \},$$
 where $C_\theta = {\cal O}(1/\theta^{1+\theta/\alpha})$ is given in \refeq{Ctheta}, then
   \beq\label{e:t:pointwise-ii}
    \P\{ {\rm criterion} (i_0) = 1 \} \ge 1- \delta.
   \eeq
 \end{theorem}

 \medskip
 We now address the algorithmic implementation of the point query
 problem and its criterion.  This is more involved now than in the case of norms and therefore we present
 a detailed argument here.  Following Indyk \cite{indyk:2000}, suppose now that the signal is only of finite precision e.g.\
$$
 f:\{1,\ldots,N\} \to \{0,1,\ldots,L\}
$$
 and suppose, moreover, that our pseudorandom numbers can only
 take values in the set $V_L :=\{p/q,\ p, q\in \{0,1,\ldots,L\},\ q\not=0\}$.  
 Let $U_j(i)$ be infinite precision independent uniform random numbers in
 $(0,1)$. We shall
base our algorithms on discretized versions of the ideal standard $\alpha-$Fr\'echet variables
 $Z_j(i):= \Phi_\alpha^{-1}(U_j(i))$, where
 $\Phi_\alpha^{-1}(y) := (\ln(1/y))^{-1/\alpha},\ y\in(0,1)$ is the
 inverse cumulative distribution function of a standard $\alpha-$Fr\'echet variable.
 Fix a small parameter
 $\gamma>0$ (to be specified), and  suppose that 
 \beq\label{e:Uj-gamma} 
  U_j(i) \in [\gamma,1-\gamma],\ \ \mbox{ for all }i =1,\ldots, N,\ j=1,\ldots,K.
 \eeq
 This is not a limitation since this
 happens with probability at least $(1-2\gamma)^{KN}$ which is at least
 $(1-\delta)$ provided that $\gamma < \delta/(4KN) = {\cal O}(
 \delta/ KN )$.  This is so because $|\ln(1-2\gamma)| \le 4 \gamma$, $\forall \gamma \in (0,1/4)$ and
 since $|\ln(1-\delta)| \ge \delta,$  $\forall \delta \in (0,1)$.

 Let now $\wtilde U_j(i)$ be a  multiple of $1/L$, nearest
 to $U_j(i)$, which is also in the set $[\gamma,1-\gamma]\cap V_L$.
 Let
%
% SHOULD THE NEXT  $Z$ HAVE A TILDE? I STOPPED READING AT THIS POINT TO
% BE ABLE ATO SAEND YOU BACAAK THE PAPER.  Answer:  No, since I want to emphasize that it 
% may not be of the type p/q, for some integers p and q.
% 
 $Z_j^*(i):= \Phi_\alpha^{-1}(\wtilde U_j(i))$ and let $\wtilde
 Z_j(i)$ be a  multiple of $1/L$
  in the set $V_L$, nearest to $Z_j^*(i)$.
 Observe that $|\wtilde Z_j(i) - Z_j^*(i)|\le 1/L$ and, as
 in Indyk \cite{indyk:2000}, by the mean value theorem,
 \beq\label{e:beta-precision}
  |Z_j(i) - Z_j^*(i)| \le \frac{1}{L}\sup_{y\in [\gamma,1-\gamma]}{\Big|}d\Phi_\alpha^{-1}(y)/dy{\Big|}
  = {\cal O}{\Big(}\frac{1}{L \gamma^{1+1/\alpha}}{\Big)}.
 \eeq
 and therefore, 
 \beq\label{e:wtilde-Z-Z}
   |\wtilde Z_j(i) - Z_j(i)| \le \beta := {\cal O}(1/ L \gamma^{1+1/\alpha}),\ \ \mbox{ for all }i=1,\ldots,N,\  j=1,\ldots,K.
 \eeq
 
 \begin{theorem}\label{t:pointwise-algorithm} Let $\epsilon \in (0,(\alpha/(\alpha+1))^{1/\alpha^2})$, 
 $\delta\in(0,1)$ and $D>0$.  Suppose that
 \beq\label{e:t:pointwise-algorithm}
  2(1 + C_\alpha) \epsilon (\sum_{1\le i \le N} f(i)^\alpha)^{1/\alpha} \le f(i_0) \le D (\Sum_{i\not = i_0} f(i)^\alpha)^{1/\alpha},
 \eeq
  where $C_\alpha = \alpha (1+1/\alpha)^{1+1/\alpha} e^{-(1+1/\alpha)}$.

  If the precision $\beta$ in \refeq{wtilde-Z-Z} is such that $\beta \le \epsilon^\alpha/(D+1)$, then there exists an algorithm, implementing 
 the point estimator $\what f(i_0)$ so that 
  \beq\label{e:t:pointwise-algorithm-c} 
   \P \{ \what f(i_0) = f(i_0) \} \ge 1 - 3\delta,
  \eeq 
  holds.  This can be done
  in space ${\cal O} (\log_2(D N/\epsilon \delta) \log_2(N)\ln(1/\delta)/\epsilon^{\alpha})$ 
  with the same order of bit--wise operations per stream item.
 \end{theorem}

 \noi The proof is given in the Appendix.

 \medskip
 The infinite precision was essential in proving that $\{{\rm criterion}(i_0) =1\}$ 
 implies $\{\what f(i_0) = f(i_0)\}$.  We cannot expect this when using real algorithms where the
 $Z_j(i)$'s have finite precision.  The following result shows, however, that there
 is nevertheless an algorithm such that $\{{\rm criterion}(i_0) =1\}$
 implies that $\{\what f(i_0) = f(i_0)\}$ holds with high 
 probability, independently of the nature of the signal $f$.

\smallskip
 \begin{theorem}\label{t:criterion-algorithm}  Let the point estimator $\what f(i_0)$ and its criterion be based on 
  a max--stable sketch in terms of the finite precision variables $\wtilde Z_j(i)$ as in 
  \refeq{wtilde-Z-Z}.  If $\beta \le C(\delta/(K^2 [\ln(NK^2/\delta)]^{1/\alpha}))$, then
  \beq\label{e:criterion-algorithm}
    \P (\{ \what f(i_0) \not = f(i_0)\}\cap \{{\rm criterion}(i_0)=1\}) \le \delta, 
  \eeq
  where the constant $C$ does not depend on the signal $f$.  The last probability bound is also valid for an
  algorithm requiring storage ${\cal O}((\log_2(N) \ln(1/\delta)/\epsilon^{\alpha})^{{\cal O}(1)})$ and
  the same order of operations per stream item.
 \end{theorem}

\noi The proof is given in the Appendix.  

 \medskip
\noi{\bf Remarks:}
\begin{enumerate}
 \item Relation \refeq{criterion-algorithm} shows that our criterion may falsely indicate that $\what f(i_0)= f(i_0)$
 only with small probability.

 \item Our point query and its criterion algorithms have features of both Las Vegas and Monte Carlo randomized algorithms.
 Namely, they give {\it exact} results, as Las Vegas algorithms do, however their computational time is fixed and sometimes (with
 low probability) they fail to give correct results.  As in Monte Carlo algorithms, the probability of getting exact results
 grows with the size of the max--stable sketch.

\end{enumerate}

\section*{Acknowledgments}  We would like to thank Anna Gilbert, Martin Strauss and Joel Tropp
for their remarks and encouragement.  Martin Strauss helped us understand better the results of Nisan 
\cite{nisan:1990}.

\section*{Appendix}

\subsection*{Background on max--stable distributions}

\begin{definition} A random variable $Z$ is said to be max--stable if, for any $a,\ b > 0$, there exist
$c>0$ and $d\in\bbR$, such that 
\beq\label{e:max-stable-def}
 \max\{aZ', bZ''\} \stackrel{d}{=} c Z + d, 
\eeq
where $Z'$ and $Z''$ are independent copies of $Z$.
\end{definition}

This definition resembles the definition of sum--stability where the operation
``max'' is the summation.  Recall that $X$ is sum--stable if for any $a,\ b >0$,
there exist $c>0$ and $d\in\bbR$, such that
$$
 aX' + b X'' \stackrel{d}{=} c X + d.
$$
Both sum--stable and max--stable distributions arise as the limit distributions when taking
sums and maxima, respectively, of independent and identically distributed (iid) random variables.
For more details on sum--stable and max--stable variables, see e.g.\ Samorodnitsky and Taqqu 
\cite{samorodnitsky:taqqu:1994book} and Resnick \cite{resnick:1987}.  We will only review in more
detail the class of $\alpha-$Fr\'echet max--stable variables.

\begin{definition} A random variable $Z$ is said to be $\alpha-$Fr\'echet, for some $\alpha>0$, if
$$
\P\{Z \le x\} = \exp\{ - \sigma^\alpha x^{-\alpha}\},\ \ \mbox{ for } x>0,
$$
and zero otherwise (for $x\le 0$), where $\sigma>0$.  If $\sigma=1$, then $Z$ is said to
be {\it standard } $\alpha-$Fr\'echet.
\end{definition}

\noi We now list some key features of the $\alpha-$Fr\'echet variables

\medskip

\noi$\bullet$ The parameter $\sigma$ plays the role of a scale coefficient.  Indeed, for all $a>0$,
$$
 \P \{ aZ \le x\} = \P\{Z\le x/a\} = \exp\{ -(a \sigma)^{\alpha} x^{-\alpha}\},\ \ x>0,
$$
and therefore $aZ$ is $\alpha-$Fr\'echet with scale coefficient $a\sigma$.

\medskip
\noi$\bullet$ One can check by using independence that \refeq{max-stable-def} holds for any $\alpha-$Fr\'echet $Z$.
More generally, let $Z, Z(1), \ldots, Z(n)$ be iid $\alpha-$Fr\'echet with scale coefficients $\sigma>0$,
and let $f(i) \ge 0$.  Then, by independence, for any $x>0$,
$$
 \P \{ \vee_{1\le i \le n} f(i) Z(i) \le x\} = \prod_{1\le i \le n} \P \{ Z_i \le x/f(i)\} = \exp\{ - \Sum_{i=1}^n 
 f(i)^\alpha \sigma^\alpha x^{-\alpha} \},
$$
and thus 
$$ 
 \xi := \bigvee_{1\le i \le n} f(i) Z(i) \stackrel{d}{=} \sigma_\xi Z,\ \ \mbox{ where } 
 \sigma_\xi = (\sum_i f(i)^\alpha)^{1/\alpha} = \|f\|_{\ell_\alpha},$$
and where $Z$ is a standard $\alpha-$Fr\'echet variable.  That is, the weighted maxima $\xi$ is an
$\alpha-$Fr\'echet variable with scale coefficient $\sigma_\xi$ equal to $\|f\|_{\ell_\alpha}$.

\medskip
{\it This last property is one motivation to consider max--stable sketches.}  The max--stable sketch defined in 
\refeq{Ej} can be viewed as a collection of independent realizations of an $\alpha-$Fr\'echet variable
with scale coefficient equal to $\|f\|_{\ell_\alpha}$. 

\medskip
\noi$\bullet$ The max--stable sketches are {\it max--linear}.  That is, if $f,\ g \in\bbR_+^N$ are two signals, then
for any $a,\ b \ge 0$, we have:
\beq\label{e:max-linear}
E_j(af \vee bg) = a E_j(f) \vee b E_j(g),\ \ \mbox{ for all } j=1,\ldots,K.
\eeq
Indeed,
$$
E_j(af \vee bg) = \bigvee_{1\le i\le N} (af(i)\vee b g(i) )Z_j(i) = a\bigvee_{1\le i\le N} f(i)Z_j(i) \vee b \bigvee_{1\le i\le N} g(i)Z_j(i)
=aE_j(f)\vee b E_j(g).
$$

\medskip
\noi$\bullet$ The $\alpha-$Fr\'echet variables are heavy--tailed.  Namely, by using the Taylor series
expansion of $1-e^{-z}$, one can show that
$$
 \P \{ Z > x \} \sim \sigma^\alpha x^{-\alpha},\ \ \mbox{ as }x\to\infty,
$$
where $a_n \sim b_n$ means $a_n/b_n \to 1,\ n\to\infty$.  Thus, the moments $\E Z^p,\ p>0$ of $Z$ are finite only if
$0<p<\alpha$.  However, when $0<p<\alpha$, these moments can be easily evaluated.  We have
\beq\label{e:Zp}
 \E Z^p = \int_0^\infty z^p d \exp\{-\sigma^\alpha z^{-\alpha}\} = \sigma^p \int_0^\infty u^{-p/\alpha} e^{-u} du = 
 \sigma^p \Gamma(1-p/\alpha),
\eeq
where in the last integral we used the change of variables $u = \sigma^\alpha x^{-\alpha}$ and where
$\Gamma(a) := \int_0^\infty u^{a-1}e^{-u} du,\ a>0$ denotes the Gamma function.

\medskip
\noi$\bullet$ One can also easily express the median ${\rm med}(Z)$ of an $\alpha-$Fr\'echet variable $Z$.  Indeed, 
$\P\{Z \le {\rm med}(Z)\} =1/2$ and by solving $\exp\{-\sigma^\alpha {\rm med}(Z)^{-\alpha}\} = 1/2$, one obtains:
\beq\label{e:medZ}
  {\rm med}(Z) = \frac{\sigma}{(\ln 2)^{1/\alpha}}.
\eeq

\medskip
In Section \ref{s:dist}, we used Relations \refeq{Zp} and \refeq{medZ} to estimate norms and distances
of signals based on their max--stable sketches.

\medskip
\noi$\bullet$ The $\alpha-$Fr\'echet variables can be easily simulated.  If $U_j,\ j\in\bbN$ are independent
uniformly distributed variables in $(0,1)$, then $Z_j := \Phi_\alpha^{-1} (U_j) = (\ln(1/U_j))^{-1/\alpha},\ j\in\bbN$ are
independent standard $\alpha-$Fr\'echet.  Indeed, for all $x>0$,
$$
 \P \{ (\ln(1/U))^{-1/\alpha} \le x \} = \P \{ \ln (1/U) \ge x^{-\alpha} \}  =
 \P \{ U \le e^{-x^{-\alpha}}\} = e^{-x^{-\alpha}}.
$$

\subsection*{Proofs for Sections \ref{s:dist} and \ref{s:dom}}

\medskip
\noi{\sc Proof of Theorem \ref{t:norm}:}  Observe that
$$
\frac{ \ell_{\alpha,r}(f)^r}{\|f\|_{\ell_\alpha}^r} = \frac{1}{\Gamma(1-r/\alpha)K} \Sum_{j=1}^K \xi_j^r, 
$$
and
$$
 \frac{ \ell_{\alpha,{\rm med}}(f)}{\|f\|_{\ell_\alpha}} = (\ln(2))^{1/\alpha} {\rm median}\{ \xi_j,\ 1\le j \le K\},
$$
where $\xi_j,\ j=1,\ldots,K$ are independent standard $\alpha-$Fr\'echet variables.

Therefore, the result in the case of the sample--median based estimator follows from Lemma 2 in Indyk \cite{indyk:2000}, for example,
since the derivative of $\Phi_\alpha^{-1}(y) = (\ln(1/y))^{-1/\alpha}$ at $y = 1/2$ is bounded.  The result in the case of the moment
estimator follows from the Central Limit Theorem, since the variables $\xi_j^r - 1/\Gamma(1-r/\alpha)$ have
zero expectations and finite variances.  $\Box$

\medskip
\noi We will provide more detailed bounds in the above proof with absolute constants in an
extended version of this paper.

\medskip
\noi{\sc Proof of Theorem \ref{t:dist}:}  We consider only $\what\rho(f,g) = \what\rho_{\alpha,r}(f,g)$.  The argument for
the estimator $\what\rho_{\alpha,{\rm med}}(f,g)$ is similar.  Suppose that $K$ is as in Theorem \ref{t:norm}, so that
with probabilities at least $(1-\delta)$, we have
$
\ell_{\alpha,r}(f)^\alpha = \|f\|_{\ell_\alpha}^\alpha(1+{\cal O}(\epsilon)),$
$\ell_{\alpha,r}(g)^\alpha = \|g\|_{\ell_\alpha}^\alpha(1+{\cal O}(\epsilon))$, and $ \ell_{\alpha,r}(f\vee g)^\alpha 
= \|f\vee g\|_{\ell_\alpha}^\alpha(1+{\cal O}(\epsilon))$. 

Thus, by the union bound of probabilities, with probability at least $(1-3\delta)$ we have
\beq\label{e:rho-1+epsilon}
 \what\rho_{\alpha,r}(f,g) = 2 \|f\vee g\|_{\ell_\alpha}^\alpha(1+{\cal O}(\epsilon)) - \|f\|_{\ell_\alpha}^\alpha(1+{\cal O}(\epsilon)) 
- \|g\|_{\ell_\alpha}^\alpha(1+{\cal O}(\epsilon)).
\eeq
Now, since 
$$\rho_\alpha(f,g) =  2 \|f\vee g\|_{\ell_\alpha}^\alpha -  \|f\|_{\ell_\alpha}^\alpha -  \|g\|_{\ell_\alpha}^\alpha \ge 
\eta \|f\vee g\|_{\ell_\alpha}^\alpha \ge \eta \max\{  \|f\|_{\ell_\alpha}^\alpha, \|g\|_{\ell_\alpha}^\alpha\},$$
we get, from \refeq{rho-1+epsilon}, that
$$
 \what\rho_{\alpha,r}(f,g) = (2 \|f\vee g\|_{\ell_\alpha}^\alpha - \|f\|_{\ell_\alpha}^\alpha - \|g\|_{\ell_\alpha}^\alpha) (1 +{\cal O}(\epsilon/\eta)).
$$
The last relation is valid with probability at least $(1-3\delta)$, which implies \refeq{dist}. $\Box$

\subsection*{Proofs for Section \ref{s:point}}

 \medskip
\noi{\sc Proof of Theorem \ref{t:pointwise}:}
  Observe that by \refeq{Ej}, \refeq{gk} and \refeq{fhat-def}, 
 $$
  \what f(i_0) = \min_{1\le j\le K} { E_j(f)/Z_j(i_0) } = \min_{1\le j\le K} 
  f(i_0)\vee Z_{j}(i_0)^{-1} \bigvee_{i\not=i_0} f(i) Z_j(i),
 $$
 where $\vee$ denotes ``max''. Now
  $\vee_{i\not = i_0} f(i) Z_j(i)$
 is independent of  $Z_j(i_0) \stackrel{d}{=} Z_j(1) $ and,
by max--stability (see Appendix), it 
equals in distribution $(\Sum_{i\not = i_0} f(i)^\alpha)^{1/\alpha}
 Z_j(2)$.  Hence
 \beq
\label{e:fZ}
  \what f(i_0) \stackrel{d}{=} f(i_0) \vee \min_{1\le j \le K} {\Big\{} Z_j(1)^{-1} 
  (\sum_{i\not = i_0} f(i)^\alpha)^{1/\alpha} Z_j(2) {\Big\}} =: f(i_0) 
  \vee \min_{1\le j\le K} c_f(i_0) Z_j(2)/Z_j(1),
 \eeq
 where $c_f(i_0) = (\sum_{i\not = i_0} f(i)^\alpha)^{1/\alpha}$.
 By using again the independence in $j$, we get
 \begin{eqnarray}\label{e:fhat=f}
  \P \{\what f(i_0) = f(i_0) \} &=& \P\{ f(i_0) \ge \min_{1\le j \le K} c_f(i_0) Z_j(2)/Z_j(1) \}
                                = 1 - \P\{ f(i_0)  < c_f(i_0) Z_1(2)/Z_1(1) \}^K  \nonumber\\
  &=& 1- \P \{ Z_1(1)/Z_1(2) <  c_f(i_0)/f(i_0)\}^K.
 \end{eqnarray}
 By Lemma \ref{l:xi/eta}, the probability in \refeq{fhat=f} equals $1/(1+ f(i_0)^\alpha/c_f(i_0)^\alpha)$, and hence
 \beq
\label{e:ff}
  \P\{ \what f(i_0) = f(i_0) \} = 1 - {\Big(}\frac{f(i_0)^\alpha}{f(i_0)^\alpha + c_f(i_0)^\alpha}{\Big)}^K 
= 1 - {\Big(} 
  {\Sum_{i\not=i_0} f(i)^\alpha \over \Sum_{i} f(i)^\alpha  } {\Big)}^K
 = 1 - {\Big(} 1 - 
  {f(i_0)^\alpha \over \|f\|_{\ell_\alpha}^\alpha } {\Big)}^K.
 \eeq
 Now  $f(i_0) > \epsilon \|f\|_{\ell_\alpha}$ implies $\P\{\what f(i_0) = f(i_0)\} \ge
 1 - (1-\epsilon^\alpha)^K$.  By choosing $\delta \ge (1-\epsilon^\alpha)^K$, we get
 $K \ge \ln(\delta)/\ln(1-\epsilon^\alpha)$. 
 Since $|\ln(1-x)| \ge x,$ for all $x\in(0,1)$, we get
 that $K \ge \ln(1/\delta)/\epsilon^\alpha \ge \ln(\delta)/\ln(1-\epsilon^\alpha)$, for all $\epsilon \in (0,1)$, which completes
 the proof of part {\bf (i)}.

 \medskip
 We now prove part {\bf (ii)} {\bf (a)}.  Observe that 
 \beq\label{e:g=xi}
  (g_{(1)}(i_0),g_{(2)}(i_0)) \stackrel{d}{=}
 (f(i_0)\vee \xi_{(1)}(i_0)), f(i_0)\vee \xi_{(2)}(i_0)), 
 \eeq
 where $\xi_{(j)}(i_0)\le \xi_{(j+1)}(i_0),\ j\le K-1$
 is the sorted sample of independent random variables 
 $\xi_j(i_0):=c_f(i_0) Z_j(2)/Z_j(1),\ j=1,\ldots,K$. Since the joint distribution of
 $\xi_{(1)}(i_0)$ and $\xi_{(2)}(i_0)$ has a density, it follows that 
 $\P\{ \xi_{(1)}(i_0) = \xi_{(2)}(i_0)\} = 0$.
  Hence, in view of \refeq{g=xi}, with probability 1, we have
\beq
\label{e:rhs}
 \{{\rm criterion} (i_0) = 1\}
\equiv \{g_{(1)}(i_0) = g_{(2)}(i_0)\} = \{f(i_0) \ge \xi_{(2)}(i_0)\}.
\eeq
The right--hand side of \refeq{rhs} occurs only if $\{\what f(i_0) = f(i_0)\}$, which completes
the proof of {\bf (ii) (a)}.
  
\medskip
We now turn to part {\bf (ii)} {\bf (b)} and 
estimate $\P\{{\rm criterion}(i_0)=1\} = \P\{ f(i_0) \ge \xi_{(2)}(i_0)\}$.  We have,
 \begin{eqnarray*}
 \P\{ f(i_0) \ge \xi_{(2)}(i_0)\} &=&
    \P\{ f(i_0) \ge c_f(i_0) Z_{j}(2)/Z_j(1),\ \mbox{ for at least two $j$'s}\}\\
                                  &=& 1 - { K \choose 0} p^K -{ K \choose 1} p^{K-1} (1-p)
                                  \ge 1 - (K+1)p^{K-1},
 \end{eqnarray*}
where $p = \P \{f(i_0) Z_1(1) < c_f(i_0) Z_1(2)\}$.  Reasoning as in
part {\bf (i)}, we get by Lemma \ref{l:xi/eta},
$p = (1- f(i_0)^\alpha/\|f\|_{\ell_\alpha}^\alpha) < 1-\epsilon^\alpha$, since $f(i_0) >
 \epsilon\|f\|_{\ell_\alpha}$.  We thus need to choose $K$'s which
 satisfy the inequality $\delta \ge (K+1)(1-\epsilon^\alpha)^{K-1} \ge (K+1) p^{K-1}$.
 For $K\ge 3$, we have $\wtilde K := K-1 \ge (K+1)/2$, and hence 
 it suffices to have $\delta \ge 2\wtilde K (1-\epsilon^\alpha)^{\wtilde K}$, or simply,
 \beq\label{e:wtilde-K}
  \wtilde K \ge \frac{\ln(\wtilde K)}{\epsilon^\alpha} + \frac{\ln(2/\delta)}{\epsilon^\alpha},
 \eeq
 where we used that $|\ln(1-\epsilon^\alpha)|\ge \epsilon^\alpha,\ \epsilon \in (0,1)$.
 Let $\theta>0$ and $\wtilde K \ge 2\ln(2/\delta)C_\theta/\epsilon^{\alpha+\theta}$, for some
 $C_\theta \ge 1$ (to be specified).  Then, since $\wtilde K \ge 2 \ln(2/\delta) C_\theta/\epsilon^{\alpha+\theta} 
 \ge 2 \ln(2/\delta)/\epsilon^{\alpha}$, it follows that \refeq{wtilde-K} holds if $\wtilde K \ge 
 2 \ln(\wtilde K)/\epsilon^\alpha$.  Since $\wtilde K^{\alpha/(\alpha+\theta)} \ge
 (2C_\theta)^{\alpha/(\alpha+\theta)}/\epsilon^\alpha$, we get
 that \refeq{wtilde-K} holds if
 $$
   \wtilde K \ge \frac{\wtilde K^{\alpha/(\alpha+\theta)}}{(2C_\theta)^{\alpha/(\alpha+\theta)}} 2\ln(\wtilde K) \ge 
 \frac{2 \ln(\wtilde K)}{\epsilon^\alpha},\ \mbox{ or if, }\ 
 (2C_\theta)^{\alpha/(\alpha+\theta)} \wtilde K^{\theta/(\alpha+\theta)} \ge 2 \ln(\wtilde K).
 $$
 The last is equivalent to
 $$
  u^{\gamma} \ge \ln(u),\ \ \mbox{ where } u
 = \wtilde K^{2^{\theta/(\alpha+\theta)}C_\theta^{-\alpha/(\alpha+\theta)}},
 $$
 and $\gamma = 
(\theta/(\alpha+\theta))2^{-\theta/(\alpha+\theta)}
 C_\theta^{\alpha/(\alpha+\theta)}$.
  We have that
 $u^{\gamma} \ge \ln (u),\ u>1,$ for all $\gamma \ge 1/e$, and thus for $C_\theta$ we obtain:
 \beq\label{e:Ctheta}
    \gamma \equiv \frac{\theta}{\alpha+\theta}2^{-\theta/(\alpha+\theta)}
C_\theta^{\alpha/(\alpha+\theta)} \ge 1/e\ \ \mbox{ or }\ \
  C_\theta \ge \frac{2^{\theta/\alpha}}{e^{1+\theta/\alpha}} 
{\Big(}1+\frac{\alpha}{\theta}{\Big)}^{1+\theta/\alpha} + 1,
 \eeq
 where we add $1$ in the last relation to ensure that $C_\theta\ge 1$.
 $\Box$

\medskip
 \noi{\sc Proof of Theorem \ref{t:pointwise-algorithm}:} We will first show that
 \beq\label{e:2.16}
  \P \{ \wtilde f(i_0) = f(i_0)\} \ge 1- 2\delta,
 \eeq
 where $\wtilde f(i_0)$  is defined as $\what f(i_0)$ but based on {\it truly independent} variables $\wtilde Z_j(i)$
 which satisfy \refeq{wtilde-Z-Z}.  Recall that Relation \refeq{wtilde-Z-Z} holds if \refeq{Uj-gamma} holds.
 Choose $\gamma>0$ and $L$ so that $\P(B)\ge 1-\delta$, where $B$ denotes the event $\{\mbox{Condition \refeq{Uj-gamma} holds}\}$. 
 As in the proof of Theorem \ref{t:pointwise}, 
 $$
  \{\wtilde f(i_0) = f(i_0)\} = \{ f(i_0) \ge \min_{1\le j\le K} \wtilde Z_j(i_0)^{-1} \vee_{i\not = i_0} f(i) \wtilde Z_j(i)\}.
 $$
 Observe that, since $B$ holds, by \refeq{wtilde-Z-Z}, 
 $\vee_{i\not=i_0} f(i) \wtilde Z_j(i) \ge \vee_{i\not=i_0} f(i) Z_j(i) + \beta s_f(i_0)$,
 where $s_f(i_0) := \max_{i\not=i_0} f(i)$.  Thus, by using also that $\wtilde Z_j(i_0) \ge Z_j(i_0) - \beta$, we get
 \beq\label{e:t:pointwise-algorithm-1.1}
 \{\wtilde f(i_0) = f(i_0)\} \supset B \cap {\Big\{} f(i_0) \ge \min_{1\le j\le K} (Z_j(i_0)-\beta)^{-1} {\Big(}
 \vee_{i\not=i_0} f(i) Z_j(i) + \beta s_f(i_0){\Big)} {\Big\}} =: B \cap C.
 \eeq
 Since $\P(B)\ge 1-\delta$ by Relation \refeq{t:pointwise-algorithm-1.1} it follows that, to establish \refeq{2.16},
 it is enough to show that $\P(C)\ge 1-\delta$, where $C$ denotes the second event in the right--hand side of 
 \refeq{t:pointwise-algorithm-1.1}.  Indeed, $\P\{\wtilde f(i_0) = f(i_0)\} \ge \P(B\cap C) \ge  1-\P(B') - \P(C') \ge 1- 2\delta$.  
 The event $C$, however, involves ideal precision Fr\'echet random variables where their corresponding $U_j(i)$'s 
 are not bound to satisfy Condition \refeq{Uj-gamma}.  We can therefore manipulate $C$ as in the proof of Theorem \ref{t:pointwise} {\bf (i)}.

 Since $\vee_{i\not=i_0} f(i) Z_j(i) \stackrel{d}{=} c_f(i_0) Z_1(2),$ where $c_f(i_0) = (\sum_{i\not=i_0} f(i)^\alpha)^{1/\alpha}$,
 by independence:
 \begin{eqnarray*}
  \P (C)  &\ge& 1 - \P \{ (Z_1(1) - \beta) f(i_0) < c_f(i_0) Z_1(2) + \beta s_f(i_0)\}^K \\
               &=& 1-{\Big(}\E \exp\{-{\Big(}\wtilde c_f(i_0) Z_1(2) + \beta(1+\wtilde s_f(i_0)){\Big)}^{-\alpha}\}
               {\Big)}^K,
 \end{eqnarray*}
 where $\wtilde c_f(i_0) = c_f(i_0) / f(i_0)$ and $ \wtilde s_f(i_0) = s_f(i_0) / f(i_0)$.
 We now bound above the expectation in the last relation.  Note that
 $\wtilde c_f(i_0)z + \beta(1+\wtilde s_f(i_0)) \le 2\wtilde c_f(i_0) z,$ for all 
 $z > \beta(1+\wtilde s_f(i_0))/ \wtilde c_f(i_0)$.  Thus, by \refeq{Frechet} and as in Relation \refeq{xi/eta} below,
 \begin{eqnarray} \label{e:1-P}
 1-\P (C)  &\le&
  {\Big(} \P\{Z_1(2) \le \beta(1+\wtilde s_f(i_0))/\wtilde c_f(i_0)\} +  
 \E \exp\{ -(2\wtilde c_f(i_0) Z_1(2))^{-\alpha} \} {\Big)}^K \nonumber \\
 &=& {\Big(} p(\beta) + \frac{1}{(2\wtilde c_f)^{-\alpha} +1} {\Big)}^K,
 \end{eqnarray}
 where $p(\beta) = \P\{Z_1(2) \le \beta(1+\wtilde s_f(i_0))/\wtilde c_f(i_0)\} =
 \exp\{- (\beta(1+\wtilde s_f(i_0))/\wtilde c_f(i_0))^{-\alpha}\}$.

 Let now $K$ be such that $(1-\epsilon^\alpha)^K \le \delta$, that is, $K = {\cal O}(\ln(1/\delta)/\epsilon^\alpha)$.
 Then, in view of \refeq{1-P}, $\P(C) \ge 1-\delta$, provided that
 $$
  p(\beta) + \frac{1}{(2\wtilde c_f(i_0))^{-\alpha} +1} \le  1-\epsilon^\alpha,
 $$
 or, equivalently, if $f(i_0)^\alpha/(f(i_0)^\alpha + 2^\alpha \sum_{i\not=i_0} f(i_0)^\alpha)
 \ge (\epsilon^\alpha + p(\beta))$.  The last inequality holds if
 $$
   f(i_0) \ge 2(\epsilon^\alpha + p(\beta))^{1/\alpha} \|f\|_{\ell_\alpha}.
 $$
 Thus, to prove \refeq{2.16}, it remains to show that $(\epsilon^\alpha + p(\beta))^{1/\alpha} \le (1+C_\alpha)\epsilon,$
 if \refeq{t:pointwise-algorithm} holds and $\beta\le \epsilon^\alpha/(D+1)$, where  
 $\beta$ is the ``precision'' parameter in \refeq{beta-precision}.  We have that $(1 + \wtilde s_f(i_0))/\wtilde c_f(i_0) \le
 1+ f(i_0)/c_f(i_0) \le 1+ D$, and thus
 $$
  p(\beta) \le \Phi_\alpha(\beta(D+1)) \le \Phi_\alpha (\epsilon^\alpha),
 $$
 since $\beta(D+1) \le \epsilon^\alpha$.  One can show that $\Phi_\alpha''(x)\ge 0,\ \forall x\in(0,(\alpha/(\alpha+1))^{1/\alpha})$,
 and hence $\Phi_\alpha(x) \le \Phi_\alpha'(1+1/\alpha)x = C_\alpha x,\ \forall x\in(0,(\alpha/(\alpha+1))^{1/\alpha})$.
 Hence, $p(\beta)\le C_\alpha \epsilon^\alpha,$ for all $\epsilon \in (0,(\alpha/(\alpha+1))^{1/\alpha^2})$, which implies
 \refeq{2.16}.

 \medskip
 Now, consider a point estimator $\what f(i_0)$, defined as $\wtilde f(i_0)$, but with $\wtilde Z_j(i)$ replaced by pseudo--random
 variables, with the same precision (i.e.\ taking values in the set $V_L = \{p/q,\ p,q=0,1,\ldots,L\}$).  We will argue that a 
 pseudo--random number generator exists so that
 \beq\label{e:f-tilde=f-hat}
  \P \{ \wtilde f(i_0) = \what f(i_0) \} \ge 1-\delta,
 \eeq 
 for some $\wtilde f(i_0)$ based on independent $\wtilde Z_j(i)$'s.  This, in view of \refeq{2.16}, would imply
 \refeq{t:pointwise-algorithm-c}.

 We first need \refeq{Uj-gamma} to hold with probability at least $(1-\delta)$ with $\gamma>0$
 such that $\beta = {\cal O}(1/L\gamma^{1+1/\alpha}) \le \epsilon^\alpha/D$.  This can be achieved by taking $L = {\cal O}(
 (DNK/\delta\epsilon)^{{\cal O}(1)})$.  Now, to ensure that \refeq{2.16} holds, it suffices to take
 $K = {\cal O}(\ln(1/\delta)/\epsilon^{\alpha})$.  Therefore, one needs $\log_2(L) = {\cal O}(\ln(DN/\delta \epsilon))$ bits to represent each
 $\wtilde Z_j(i)$. 

 Hence our algorithm uses ${\cal O}(\log_2(L)N\ln(1/\delta)/\epsilon^{\alpha})$ random bits.  As in Indyk \cite{indyk:2000}, by using 
 the results of Nisan \cite{nisan:1990}, for each $j,\ 1\le j\le K,$ one can generate pseudo--random $\overline{U_j}(i)$'s, which are
 ``very close'' to some independent $\wtilde U_j(i)$'s by using only
 ${\cal O}(\log_2(L)\log_2(N/\delta)\ln(1/\delta)/\epsilon^{\alpha})$ truly random seeds.  
  These $\overline{U}_j(i)$'s would ``fool'' our algorithm with probability
 at least $(1-\delta)$  since it uses only ${\cal O}( \log_2(L) \ln(1/\delta)/\epsilon^{\alpha})$ space for computations with random bits.
 That is, one has \refeq{f-tilde=f-hat}.  In summary, we need to store ${\cal O}(K \log_2(L)\log_2(N/\delta))
 = {\cal O}(\log_2(DN/\delta \epsilon) \log_2(N/\delta) \log_2(1/\delta)/\epsilon^\alpha)$ bits, needed primarily for the truly random seeds,
 and to perform about ${\cal O}(K\log_2(L)) = {\cal O}( \log_2(DN/\delta \epsilon) \log_2(1/\delta)/\epsilon^\alpha)$ 
 of bit--wise operations per each stream item, in order to maintain the sketch.
$\Box$

\medskip
\noi{\sc Proof of Theorem \ref{t:criterion-algorithm}:}
Let, as in \refeq{gk}, $\wtilde g_j(i_0):= \vee_{i} f(i) \wtilde Z_j(i)/\wtilde Z_j(i_0) =: f(i_0) \vee \wtilde \xi_j(i_0)$, where
$\wtilde \xi_j(i_0) := \vee_{i\not=i_0} f(i) \wtilde Z_j(i)/\wtilde Z_j(i_0)$, $j=1,\ldots,K$.  Since
 $\{{\rm criterion}(i_0)=1\} = \{f(i_0) \vee \wtilde \xi_{(2)}(i_0) = f(i_0) \vee \wtilde \xi_{(1)}(i_0)\}$, it follows that
$$
\{\what f(i_0) \not = f(i_0)\} \cap \{ {\rm criterion}(i_0)=1\}  \subset \{ \wtilde \xi_{(2)}(i_0) =  \wtilde \xi_{(1)}(i_0) \},
$$
where $\wtilde \xi_{(1)}(i_0) \le \wtilde \xi_{(2)}(i_0) \le \cdots $ is the ordered sample of $\wtilde \xi_{j}(i_0),\ j=1,\ldots,K$.

Therefore, the probability in \refeq{criterion-algorithm} is bounded above by:
\begin{eqnarray}\label{e:crit-alg-1}
 \P\{ \wtilde \xi_{(2)}(i_0) =  \wtilde \xi_{(1)}(i_0) \} &\le& \P\{ \wtilde \xi_{j_1} = \wtilde \xi_{j_2},\ \mbox{ for some }j_1\not=j_2,\ 
 j_1,\ j_2=1,\ldots, K\}\nonumber\\
 &\le& {K \choose 2} \P \{\wtilde \xi_{1}(i_0) =  \wtilde \xi_{2}(i_0) \}.
\end{eqnarray}
We now focus on bounding the last probability.  Since the $\wtilde \xi_j(i_0)$'s  are independent and discrete random variables, we have
\beq\label{e:sum-atoms}
 \P \{\wtilde \xi_{1}(i_0) =  \wtilde \xi_{2}(i_0) \} = \Sum_{x \, :\, \mbox{ $x$ is an atom}} \P \{\wtilde \xi_1(i_0) = x \} ^2.
\eeq
Let $\eta>0$ and observe that $(z-\beta) \ge (1-\eta)z$ and $(z+\beta)\le (1+\eta)z,$ for all $z\ge \beta/\eta$.  Thus, in view of
\refeq{wtilde-Z-Z}, 
$$
 \wtilde \xi_1(i_0) \le \vee_{i\not=i_0} f(i) \frac{(Z_1(i) + \beta)}{(Z_1(i_0) - \beta)} \le \vee_{i\not=i_0} f(i) \frac{(1+\eta)}{(1-\eta)} 
 \frac{Z_1(i)}{Z_1(i_0)},\ \ \mbox{ if } Z_1(i) \ge \beta/\eta,\ \forall i.
$$
Since the $Z_1(i)$'s are ideal precision, independent and $\alpha-$Fr\'echet, 
$\vee_{i\not=i_0} f(i)\frac{(1+\eta)}{(1-\eta)} Z_1(i)/Z_1(i_0) \stackrel{d}{=} c_f(i_0)\frac{(1+\eta)}{(1-\eta)} Z_1(2)/Z_1(1),$
where $c_f(i_0) = (\sum_{i\not=i_0} f(i)^\alpha)^{1/\alpha}$.  Therefore,
$$
 \wtilde \xi_1(i_0) \stackrel{d}{\le} c_f(i_0) \frac{(1+\eta)}{(1-\eta)} \frac{Z_1(2)}{Z_1(1)} =: \xi^*,\ \ \mbox{ if }Z_1(i) \ge \beta/\eta,\ 
 \forall i,
$$
where $\stackrel{d}{\le}$ denotes dominance in distribution.  Similarly $\wtilde \xi_1(i_0) \stackrel{d}{\ge} \xi_*$, where
$\xi_*:= c_f(i_0) \frac{(1-\eta)}{(1+\eta)} Z_1(2)/Z_1(1)$.  Thus,
\begin{eqnarray*}
\P\{ \wtilde \xi_1(i_0) = x\} &\le& \P(\{\xi_1(i_0) = x\} \cap \{Z_1(i) \ge \beta/\eta,\ \forall i\}) 
 +   \P\{Z_1(i) \not\ge \beta/\eta,\ \mbox{ for some }i\}\\
& \le & F_{\xi_*}(x) - F_{\xi^*}(x) + (1-\P\{Z_1(1) > \beta/\eta\}^N),
\end{eqnarray*}
where the last inequality follows from the fact that $\xi_*$ and $\xi^*$ have continuous cumulative
distribution functions $ F_{\xi_*}(x) := \P\{\xi_*\le x\}$.

Thus, for the probability in \refeq{sum-atoms}, we get:
\beq\label{e:crit-alg-2}
 \P \{\wtilde \xi_{1}(i_0) =  \wtilde \xi_{2}(i_0) \} \le \sup_{x >0} (F_{\xi_*}(x) - F_{\xi^*}(x)) + 
 (1-(1-\exp\{-(\beta/\eta)^{-\alpha}\})^N).
\eeq
By taking $\beta/\eta \le (1/\ln(NK^2/\delta))^{1/\alpha}$, we can make the second term in the right--hand side of \refeq{crit-alg-2}
smaller than $\delta/K^2$ (Lemma \ref{l:crit-alg-1}).  Indeed, the second term is a monotone increasing function of
$(\beta/\eta)$.  Hence by setting $\epsilon:= (1/\ln(NK^2/\delta))^{1/\alpha}$ the upper bound in \refeq{l:crit-alg-1} becomes
$N \exp\{ - (1/[\ln(NK^2/\delta)]^{1/\alpha})^{-\alpha}\} = N \exp\{ -\ln(NK^2/\delta)\} = \delta/K^2$.
Also, by Lemma \ref{l:crit-alg-2}, the first term in the right--hand side of \refeq{crit-alg-2}, is bounded above by
$\alpha 2^{4\alpha+3}\eta$, for all $\eta \in (0,1/2)$.  Thus, in view of \refeq{crit-alg-1}, the probability in
\refeq{criterion-algorithm} is bounded  above by ${\cal O}(K^2\eta) + \delta/2$, which can be made
 smaller than $\delta$ by taking $\eta = {\cal O}(\delta/K^2)$ and $\beta/\eta = {\cal O}((1/\ln(NK^2/\delta))^{1/\alpha})$.  This 
implies that $\beta = {\cal O}(\delta/(K^2 [\ln(NK^2/\delta)]^{1/\alpha}))$ would ensure that \refeq{criterion-algorithm} holds.  
Observe that the constant in the last ${\cal O}-$bound does not depend on the signal $f$.
$\Box$

\subsection*{ Auxiliary lemmas}

\begin{lemma}\label{l:xi/eta} Let $\xi$ and $\eta$ be independent, standard $\alpha-$Fr\'echet variables.  Then, for all $x>0$,
$$
 \P \{ \xi/\eta \le x \} = \frac{1}{x^{-\alpha}+1}.
$$
\end{lemma}
\begin{proof}  By independence, and in view of \refeq{Frechet}, we have:
\beq\label{e:xi/eta}
 \P\{\xi/\eta\le x\} = \E \exp\{-(\eta x)^{-\alpha}\} = \int_0^\infty e^{-y^{-\alpha} x^{-\alpha}} d e^{-y^{-\alpha}}
= \int_0^1 u^{x^{-\alpha}} du = \frac{1}{x^{-\alpha}+1}.
\eeq
Here, we used the change of variables $u:=e^{-y^{-\alpha}}$.
$\Box$
\end{proof}

\begin{lemma}\label{l:crit-alg-1} For all $\epsilon \in (0,1)$ and $N\ge 1,$ we have
\beq\label{e:l:crit-alg-1}
  1-(1-\exp\{ -(\epsilon)^{-\alpha}\})^N \le  Ne^{-\epsilon^{-\alpha}}.
\eeq
\end{lemma}
\begin{proof}
We have that $(1-x)^N \ge 1- N x,$ for all $x \in (0,1)$.  Thus, for all $x\in(0,1)$,
$1 - (1-x)^N \le 1- (1-Nx) = Nx$  and by setting $x:= e^{-\epsilon^{-\alpha}}$, we obtain
\refeq{l:crit-alg-1}. $\Box$
\end{proof}

\begin{lemma}\label{l:crit-alg-2} Let $\xi^* = c\frac{(1+\eta)}{(1-\eta)} Z'/Z''$
 and $\xi_* = c\frac{(1-\eta)}{(1+\eta)} Z'/Z''$, for some $c>0$ and $\eta\in(0,1/2)$, where
 $Z'$ and $Z''$ are independent standard $\alpha-$Fr\'echet variables.  Then, we have
$$
 \sup_{x>0} |F_{\xi_*}(x) - F_{\xi^*}(x)| \le \alpha 2^{4\alpha+3} \eta,\ \ \ \forall \eta\in(0,1/2).
$$
where $F_{\xi_*}(x):=\P\{\xi_*\le x\}$ and  $F_{\xi^*}(x):=\P\{\xi^*\le x\}$.
\end{lemma}
\begin{proof} We have that
$$
 \Delta F := F_{\xi_*}(x) - F_{\xi^*}(x) = \P\{ Z'/Z'' \le C_* x\} -  \P\{ Z'/Z'' \le C_* x\},
$$
where $C_* = (1+\eta)/c(1-\eta)$ and $C^* = (1-\eta)/c(1+\eta)$.  By Lemma \ref{l:xi/eta}, we have that
$$
\Delta F = \psi{\Big(}\frac{(1-\eta)^\alpha}{(1+\eta)^\alpha x^\alpha}{\Big)} - 
\psi{\Big(}\frac{(1+\eta)^\alpha}{(1-\eta)^\alpha x^\alpha}{\Big)},
$$ 
where $\psi(y) = 1/(c^{-\alpha}y + 1)$.  By the mean value theorem, since $|\psi'(y)|= c^{-\alpha}/(c^{-\alpha}y+1)^2$ 
is monotone decreasing in $y>0$, the last expression is bounded above by:
\begin{eqnarray}\label{e:l:crit-alg-2-1}
 & & |\psi'(\frac{(1-\eta)^\alpha}{(1+\eta)^\alpha x^\alpha})| {\Big |} \frac{(1+\eta)^\alpha}{(1-\eta)^\alpha x^\alpha}
 - \frac{(1-\eta)^\alpha}{(1+\eta)^\alpha x^\alpha} {\Big|}
\nonumber\\
& & \ \ \ \ \ \ \ \ \ \ \  = \frac{c^{-\alpha}(1+\eta)^{2\alpha}x^\alpha}{(c^{-\alpha}(1-\eta)^\alpha + (1+\eta)^\alpha x^\alpha)^\alpha}
{\Big(} {(1+\eta)^{2\alpha} - (1-\eta)^{2\alpha} \over (1-\eta^2)^{\alpha} } {\Big)}\nonumber\\
& & \ \ \ \ \ \ \ \ \ \ \ \ \le \frac{(1+\eta)^\alpha}{2(1-\eta)^\alpha} {\Big(}{ (1+\eta)^{2\alpha} - (1-\eta)^{2\alpha}
\over (1-\eta^2)^\alpha} {\Big)} = { (1+\eta)^{2\alpha} - (1-\eta)^{2\alpha} \over 2(1-\eta)^{2\alpha}},
\end{eqnarray}
where the last inequality follows from the fact that $ab/(a+b)^2 \le 1/2$, $a,\, b\in \bbR$, with
$a:= c^{-\alpha}(1-\eta)^\alpha$ and $b=(1+\eta)^\alpha x^\alpha$.  By using the mean value theorem again, we obtain that,
for some $\theta\in[-1,1],$ the right--hand side of \refeq{l:crit-alg-2-1} is bounded above by 
$$
{2\alpha (1+\theta \eta)^{2\alpha-1}  \over (1-\eta)^{2\alpha} }2\eta \le \alpha 2^{2\alpha +2 + |2\alpha-1|} \eta \le
 \alpha 2^{4\alpha+3} \eta,
$$
for all $\eta\in(0,1/2)$.  $\Box$
\end{proof}

%\bibliography{/afs/umich.edu/user/s/s/sstoev/doc/articles/mst-bibfile}
%\bibliography{/home/sstoev/doc/articles/mst-bibfile}

\end{document}